\documentclass[11pt]{amsart}
\usepackage{latexsym}
\usepackage{amsfonts}
\usepackage{hyperref}

\usepackage[top=1in, bottom=1.2in, left=1in, right=1in]{geometry}

\newtheorem{example}{Example}

\newtheorem{proposition}{Proposition}

\begin{document}

\author[]{Mariya Bessonov}
\address{Department of Mathematics, New York City College of Technology, CUNY} \email{mbessonov@citytech.cuny.edu}

\author[]{Dima Grigoriev}
\address{CNRS, Math\'ematiques, Universit\'e de Lille, 59655, Villeneuve d'Ascq, France}
\email{Dmitry.Grigoryev@univ-lille.fr}

\author[]{Vladimir Shpilrain}
\address{Department of Mathematics, The City  College  of New York, New York,
NY 10031} \email{shpil@groups.sci.ccny.cuny.edu}

\title{Probability theory and public-key cryptography}


\begin{abstract}
In this short note, we address a common misconception at the interface of probability theory and public-key cryptography.

\end{abstract}

\maketitle

\section{Introduction}

Suppose Alice sends a secret bit $c$ to Bob over an open channel, in the presence of a
computationally unbounded (passive) adversary Eve. Let $P_B$ and $P_E$ be the probabilities for Bob and Eve,
respectively, to correctly recover $c$. It is well known that if $P_B=1$, then a straightforward encryption emulation attack gives $P_E=1$ as well. In this note, we address a common misconception that this can be somehow ``generalized" to the case $P_B<1$, i.e., emulating encryption (or receiver's algorithms, or both) can give $P_E = P_B$. The basic reason why this is wrong is that Eve can never fully emulate Alice or Bob since Eve's probability space is inherently different from that of Alice or Bob. When $P_B=1$, this does not matter because ``always correct" in a probability space implies ``always correct" in any probability subspace. However, if $P_B<1$, the situation can be quite different.

In \cite{framework}, we have gone where no cryptographer had gone before  and suggested that it might be possible to build a public-key cryptographic protocol entirely based on probability theory, without using any algebra or number theory. This was met with skepticism (to put it mildly)  based on a strong belief in impossibility of having $P_B > P_E$. This skepticism has materialized in a preprint by Panny \cite{Panny} who courageously delved into the depths of elementary probability theory  and tried to actually compute some probabilities instead of just saying ``this is impossible because this cannot possibly be possible" as most other believers in ``flat Earth" do. His preprint is in two independent parts: theoretical, where he does probability computations attempting to prove $P_E \ge P_B$,  and (completely unrelated) experimental part where he offers a statistical attack on ciphertext in our protocol in \cite{framework} making Eve succeed (in recovering Alice's secret bit) with an unspecified probability $P_E > \frac{1}{2}$.
The fact that this probability (or, rather, an experimental approximation thereof) was not specified is unfortunate since it leaves open the question of whether or not this particular attack yields $P_E \ge P_B$ for the protocol in \cite{framework}.

The main purpose of this short note is to show (in Section \ref{whyfail}) that there is no ``generic" algorithm (like emulating encryption, or receiver's algorithms, or both) for Eve  to guarantee $P_E \ge P_B$. Of course, for any particular protocol, there might be an ``intelligent", protocol-specific attack, that might give $P_E \ge P_B$, but the question of whether or not there is {\it always} a  protocol-specific attack that succeeds with probability $P_E \ge P_B$ remains open.

\smallskip

\noindent For the record:

$\bullet$ We admit that in our scheme in \cite{framework}, $P_E > \frac{1}{2}$. We explain in Section \ref{how} below why and how Eve can achieve that.

$\bullet$ The $P_B > P_E$ claim for the general, ``framework", scheme in \cite{framework} still stands. In Section \ref{how}, we reproduce this framework and give some argument in support of this claim.

$\bullet$ In Section \ref{whyfail}, we give an explanation of why typical  ``proofs"  of $P_E = P_B$ (or even of  $P_E \ge P_B$) are flawed. Then we specifically address the decryption emulation attack, to answer a popular concern along the lines of ``if Eve is computationally unbounded, she can just emulate Bob and be at least as successful as Bob is in recovering Alice's plaintext". Here ``decryption emulation attack" is a slang for emulating all the receiver's algorithms used in a protocol.
\smallskip

Section \ref{how} also explains why in schemes like the one in \cite{framework}, $P_B$ inherently cannot be larger than 0.75. Regretfully, this probability appears to be not large enough to be useful in any meaningful real-life scenario, as far as we can see.

\smallskip

Finally, we encourage curious readers to read about the famous Monty Hall problem \cite{Monty}, to appreciate the importance of the probability space, and not just ``random coins", in computing probabilities. A quote from Wikipedia \cite{Monty}:
``Paul Erd\"os, one of the most prolific mathematicians in history, remained unconvinced until he was shown a computer simulation..." shows that non-believers in a (sometimes crucial) role of the probability space are in a good company.

We realize that firm believers in ``flat Earth" will not even read our note because it is much  easier to accuse of heresy than to search for the truth, but we hope that more open-minded readers will be curious enough to find out how  probability theory just a little bit beyond the first course in discrete mathematics can be used in cryptographic constructions.

\section{How $P_B$ can possibly be larger than $P_E$: a generic example} \label{how}

Let Alice be the sender of a secret bit $c$ and Bob the receiver. Suppose Alice has two disjoint probability spaces, $S_1$ and $S_2$, to pick her encryption key from. Assume, for simplicity of the analysis, that $S_1$ and $S_2$ are public (although they are typically not) and that Alice will select between $S_1$ and $S_2$ with probability $\frac{1}{2}$ (although this probability may be private as well).

Suppose that if Alice picks her encryption key from $S_1$, then Bob decrypts correctly with probability $q_1$, and if she picks her encryption key from $S_2$, then Bob decrypts correctly with probability $q_2$. Then Bob decrypts correctly with probability $\frac{1}{2}(q_1+q_2)$. Suppose $q_2>q_1$ and $q_2+q_1>1$. The latter condition implies that, in some instances, an encryption key from $S_1$ produces the same ciphertext as some encryption key  from $S_2$ does. Denote by $S_{12} \subseteq S_1$ the set of these ``special" encryption keys.

Let $\tau_1$ be the probability of the following event: Alice picked an encryption key from $S_{12}$,
conditioned on (Alice picked an encryption key from $S_1$ and Bob decrypted correctly). Why do we need this weird-looking condition? It is needed to express, in terms of $q_1$, the probability for Bob to decrypt correctly in case Alice picked an encryption key from $S_{12}$ (after choosing to pick it from $S_1$). Indeed, this probability is equal to $\tau_1 q_1$ by the probability of the intersection of two events formula. The two events here are (both conditioned on Alice having picked her encryption key from $S_1$):
(1) Alice picked her encryption key from $S_{12}$; (2) Bob decrypted  correctly.


How is Eve going to decrypt? The most obvious way is to narrow down the selection of decryption key (while emulating Bob's decryption algorithm) by assuming that Alice has picked her encryption key from
$S_2$ (since $q_2>q_1$ gives Bob a better chance for success in that case). Then, Eve would emulate Bob's algorithm in the hope that this will give her the correct decryption of Alice's bit with probability $q_2 > \frac{1}{2}(q_2+q_1)$. However, since Alice selects $S_2$ with probability $\frac{1}{2}$, the actual probability for Eve to decrypt Alice's bit correctly (if she uses this strategy) is $P_E = \frac{1}{2}q_2 + \frac{1}{2} \tau_1 q_1$. Here $q_2$ is the probability for Eve to decrypt correctly (by emulating Bob's randomness) in case Alice selected  $S_2$ to pick her encryption key from, and $\tau_1 q_1$ is the probability for Eve to decrypt correctly in case Alice selected  $S_1$ (see above). Then we have:
\smallskip

$\bullet$ The probability for Bob to decrypt Alice's bit correctly is $P_B = \frac{1}{2}(q_2+q_1)$.
\medskip

$\bullet$ The probability for Eve to decrypt Alice's bit correctly (if she uses the above strategy) is $P_E = \frac{1}{2}(q_2 + \tau_1 q_1)$. Thus, if $q_2 + \tau_1 q_1 > 1$, then $P_E > \frac{1}{2}$.
\medskip

$\bullet$  $P_E < P_B$. This is because $\tau_1 q_1 < q_1$. Indeed, obviously one cannot have $\tau_1 q_1 > q_1$, and $\tau_1 q_1 = q_1$, or  $\tau_1=1$, would  defy the purpose for Alice to have a separate $S_1$ in the first place.
\medskip

Thus, this most obvious attack does not give $P_E \ge P_B$ if Alice is able to select $S_1$ and $S_2$ such that $q_2>q_1$, $q_2+q_1>1$,  and $\tau_1 < 1$. An example of such selection was given in \cite{framework}. It is straightforward to see that other strategies (i.e., other probability distributions) for Eve to select between $S_1$ and $S_2$ for a supposed  encryption key will result in an even lower probability $P_E$ of success.

In particular instantiations of this general idea there might be instantiation-specific statistical attacks on Bob's public key or Alice's ciphertext \cite{Panny}, but the point we are trying to make here is that, contrary to what skeptics claim, there is no ``universal" (e.g. encryption/decryption emulation) attack on such a scheme that would guarantee $P_E \ge P_B$. We will establish this more formally in the next section.

To conclude this section, we note that we were unable to find an instantiation of this general scheme where  both $q_2$ and $q_1$ would be greater than $\frac{1}{2}$, so it appears that $P_B = \frac{1}{2}(q_2+q_1) < \frac{3}{4}$ in any instantiation of this scheme. In the instantiation offered in \cite{framework}, $P_B$ is approximately 0.55.


\section{Why all ``proofs" of $P_E = P_B$ fail} \label{whyfail}

Below is a short version of a typical ``proof" of $P_E = P_B$. 
In what follows, Alice is the sender of a secret plaintext $K_A$ and Bob the receiver who, upon decrypting Alice's ciphertext, obtains $K_B$ and wants $K_B=K_A$, with probability $P_B>\frac{1}{2}$. The adversary Eve wants to recover $K_A$, with probability $P_E \ge P_B$. Our main goal in this section is to show that {\it emulation attacks} (be it emulation of encryption, or decryption, or both) cannot give $P_E = P_B$ in any meaningful instantiation of the general scheme from  Section \ref{how}, including the one in \cite{framework}. First we briefly reproduce a typical claim, with a ``proof".

\begin{proposition} \label{emulation} Let $R_A$ be Alice's randomness, $R_B$ Bob's randomness, and $T$ the (public) transcript of communication. Suppose $R_A$ conditioned on $T$ and $R_B$ conditioned on $T$ are independent. Let $K_A$ be Alice's plaintext, $K_B$ the result of Bob's decryption, and $p=P_B$ the probability of having $K_B=K_A$ after the communication protocol execution. Then unbounded Eve, on input $T$, can generate a value $K_C$ such that $P_E = Pr(K_C = K_A) = p$.

\end{proposition}

\begin{proof} Let $K_A = f(R_A, T)$ and $K_B = g(R_B, T)$.
Conditioned on $T$, Eve can sample Bob's coins. Let $R_B'$ denote Bob's randomness emulated by Eve.
Output the value $K_C = g(R_B', T)$, which is what Bob would output on input $(R_B', T)$.
The triples $(R_A, R_B, T)$ and $(R_A, R_B', T)$ are identically distributed. Hence the values of $p$ are identical.

\end{proof}

Below we point out some issues with this proof that show that the proof is, at the very least, incomplete if $p<1$. If $p=1$, the claim of the proposition is well known to be true, as established by a straightforward encryption emulation attack.


We note, in passing, that $R_A$ and $R_B$ include not only ``random coins", but also probability spaces. Random coins of Alice and Bob are, indeed, independent in any meaningful public-key communication model.
Probability space of the sender, on the other hand, can be dependent on the receiver's public key and therefore on his randomness; this happens even in some well-established schemes, e.g. in PollyCracker.
This is not a serious issue though, just something to keep in mind.
\medskip




\noindent {\bf Serious issue.} Assume, for the sake of argument, that the claim ``The triples $(R_A, R_B, T)$ and $(R_A, R_B', T)$ are identically distributed" in the above proof is correct under appropriate independence conditions. Even that, however, does not prove the claim of the proposition, which is: ``Then unbounded Eve, on input $T$, \underline{can generate} value $K_C$ such that $Pr(K_C = K_A) = Pr(K_B = K_A) = p$".  How can Eve do that? Assume for simplicity that $K_A$ is just a single bit.

What the above proof suggests is basically an ``encryption/decryption emulation attack". That is, Eve  generates all possible Alice's (plaintext, ciphertext) pairs and all possible Bob's decryption keys, with all possible randomness, that would match the public key and the protocol description. Then Eve selects all
(plaintext, ciphertext, decryption key) triples that give $K_B=K_A$. (Recall that $K_A$ is Alice's plaintext and $K_B$ is the result of decrypting Alice's ciphertext by Bob.) Some of these triples will have $K_B=K_A=0$, while others will have $K_B=K_A=1$.   Then what? Select a triple from this pool uniformly at random (or using whatever other distribution)? Then Eve's probability space will be very different from Bob's, and therefore there is no reason for $P_E$ to be equal to $P_B$ with this strategy.

Thus, the above proof is at least incomplete since it does not mention any algorithm for Eve to make that choice and \underline{actually generate} a value $K_C$ that would be equal to $K_A$ with probability $P_B$. \qed
\medskip

To be fair, our argument above only shows that there is no algorithm for Eve to achieve $P_E=P_B$. But what about $P_E\ge P_B$? To try to achieve this, the best strategy for Eve is probably to forget about Bob's algorithms, emulate just Alice's encryption algorithm, create a probability distribution on the set of all possible (plaintext, ciphertext) pairs, for all possible values of Bob's public key, and then, when given a ciphertext, select the plaintext that corresponds to it with higher probability.
This basically takes us to the situation considered in Section \ref{how}: this strategy will guarantee $P_E>\frac{1}{2}$, but $P_E\ge P_B$ is still questionable because Alice's (private) probability space is narrower than Eve's. To illustrate how this matters, here is a simple

\begin{example}\cite{Boy}
In a city where every family has two children, Alice and  Eve walk down the street and meet Bob with a little boy in a stroller; this boy is Bob's public key. Bob tells them that he has two children, but the older child (Bob's private key) is at school now, and Bob suggests that Alice and Eve try to guess whether the other child is a boy or a girl. While Eve walks away building a probability distribution on the set of all possible gender pairs, Alice finds out that the boy in the stroller was born on  a Tuesday. This did not give Alice any information about the other child's gender, but it changed Alice's probability space! Now it is not ``all  families with two children where the younger child is a boy" but, say, ``all  families with two children where the younger child is a boy and with a boy born on a Tuesday".

The result is: Eve comes to the conclusion that the other child is a boy with probability $\frac{1}{2}$ (because the pairs (GG), (BG), (GB), (BB) are assumed to be equally likely in a random family with two children, and the fact that the boy in the stroller is the younger child narrows it down to (GB), (BB)), whereas Alice comes to the conclusion (using Bayes' formula) that the other child is a boy with probability $\frac{13}{20}>\frac{1}{2}$.

\end{example}

One can say that in this example, Alice got information not available to Eve (even though this information is irrelevant to Bob's private key), and this seems to be prohibited by theoretical cryptography rules of engagement (a.k.a. Kerckhoffs's principles). However,  in actual cryptographic scenarios (including the one in \cite{framework}), Alice can ``artificially" change her own probability space to her liking. Eve, of course, is aware of all possible probability space choices by Alice, but all she can do is ``average out" their probability distributions, which will almost for sure result in different probability distributions on the set of (plaintext, ciphertext) pairs for Alice and for Eve; sometimes it may even reverse the preference of one plaintext (given a ciphertext) over another.
This phenomenon is called {\it Simpson's paradox} \cite{stats}:
{\it a trend can appear in several different groups of data but disappear or reverse when these groups are combined}. In reference to the above example, Alice could use any information including information available also to Eve (e.g. the boy in the stroller is blonde) to (privately) narrow down her probability space, while Eve will not know which probability space Alice has chosen. Compare this to the general scheme in Section \ref{how}.

\medskip

Finally, we consider the  attack where Eve emulates just Bob (the receiver). If the probability distribution used by Bob to generate his public key is known to the public (which was the case in \cite{framework}), then decryption emulation attack may seem like a reasonable strategy for Eve, i.e., Eve can generate all possible Bob's private keys, then generate all possible Bob's public keys corresponding to each of his private keys, and then select all (private key, public key) pairs with public key matching the one actually published by Bob. This will yield a probability distribution (conditioned on $T$) on the set of all possible Bob's private keys, and Eve can select one of the private keys that occurs with highest probability in this distribution. The  probability $P_E=Pr(K_C=K_B)$ might then be larger than $\frac{1}{2}$, but this probability has little to do with $P_B=Pr(K_B=K_A)=p$ since the latter probability is largely controlled by Alice. Below we show that if Eve achieves $P_E=Pr(K_C=K_B)>\frac{1}{2}$, then, in fact, $P_E<P_B$ provided $P_B>\frac{1}{2}$.

Emulating Bob will result in the following success probability $P_E$ for Eve to recover $K_A$ in the case where $K_A$ is a single bit (assuming that $K_B$ and $K_C$, too, can only take  values 0 or 1):

$$P_E = Pr(K_C=K_A) = Pr(K_C=K_B) \cdot Pr(K_A=K_B) + Pr(K_C \ne K_B) \cdot Pr(K_A \ne K_B).$$

\noindent All probabilities here are conditioned on $T$. Also, we assume that, since Eve emulates just Bob, Eve's and Alice's randomness are independent
(conditioned on $T$), hence the events $K_C=K_B$ and $K_A=K_B$ (conditioned on $T$) are independent.

 Denote $Pr(K_B=K_A)=p, ~Pr(K_C=K_B)=\sigma$.  Then we have:

$$P_E = Pr(K_C=K_A) = \sigma p + (1-\sigma)(1-p) = \frac{1}{2}(2\sigma - 1)(2p-1) + \frac{1}{2}.$$

If $p>\frac{1}{2}$ and $\sigma>\frac{1}{2}$, then $P_E >\frac{1}{2}$, but we claim that $P_E$  is less than $p$ in this case. Indeed, $\frac{1}{2}(2\sigma - 1)(2p-1) + \frac{1}{2} < p$ is equivalent to $(2\sigma - 1)(2p-1) < 2p-1$, which is true since $\sigma < 1$. The latter holds because in a scenario similar to that in Section \ref{how}, typically (in particular, in the scheme in \cite{framework}), to the same Bob's public key, any private (decryption) key from the pool of all private keys can be associated with nonzero probability. In particular, there will be private keys that yield $K_B=1$, as well as those that yield $K_B=0$.

Thus, while there might be an {\it  instantiation-specific statistical attack} on Bob's public key, it will have nothing to do with emulation attack(s) suggested by the above proof of Proposition \ref{emulation}.
This also explains why the two parts (theoretical and experimental) in \cite{Panny} are completely unrelated and perhaps also why the probability $P_E$ in the experimental part of \cite{Panny} is not specified.

\end{document}